\documentclass[submission,copyright,creativecommons]{eptcs}
\providecommand{\event}{QPL 2015}
\usepackage{breakurl}

\usepackage{amsmath}
\usepackage{amssymb}
\usepackage{amsthm}
\usepackage{mathtools} %% For Coloneqq
\usepackage{stmaryrd}  %% For llbracket, rrbracket

\newcommand{\N}{\mathbb{N}}

\newtheorem{thm}{Theorem}[section]
\newtheorem{prp}[thm]{Proposition}

\theoremstyle{definition}
\newtheorem{dfn}[thm]{Definition}
\newtheorem{exm}[thm]{Example}
\newtheorem{ntt}[thm]{Notation}

\newcommand{\sskip}{\texttt{skip}}
\newcommand{\scomp}[2]{#1 \texttt{;} #2}
\newcommand{\sif}[3]{\texttt{if~} #1 \texttt{~then~} #2 \texttt{~else~} #3 \texttt{~fi}}
\newcommand{\swhile}[2]{\texttt{while~} #1 \texttt{~do~} #2 \texttt{~od}}

\newcommand{\tX}{\texttt{X}}
\newcommand{\tY}{\texttt{Y}}
\newcommand{\tZ}{\texttt{Z}}
\newcommand{\tH}{\texttt{H}}
\newcommand{\tS}{\texttt{S}}
\newcommand{\tT}{\texttt{T}}
\newcommand{\tCX}{\texttt{CX}}
\newcommand{\tif}{\texttt{if}}
\newcommand{\tthen}{\texttt{then}}
\newcommand{\telse}{\texttt{else}}
\newcommand{\tfi}{\texttt{fi}}

\newcommand{\sX}[1]{\tX\texttt{(}#1\texttt{)}}
\newcommand{\sY}[1]{\tY\texttt{(}#1\texttt{)}}
\newcommand{\sZ}[1]{\tZ\texttt{(}#1\texttt{)}}
\newcommand{\sH}[1]{\tH\texttt{(}#1\texttt{)}}
\newcommand{\sS}[1]{\tS\texttt{(}#1\texttt{)}}
\newcommand{\sT}[1]{\tT\texttt{(}#1\texttt{)}}
\newcommand{\sCX}[2]{\texttt{CX(}#1\texttt{,} #2\texttt{)}}
\newcommand{\sem}[1]{\llbracket #1 \rrbracket}
\def\bra#1{\mathinner{\langle{#1}|}}
\def\ket#1{\mathinner{|{#1}\rangle}}
\newcommand{\ketbra}[2]{\ket{#1}\!\!\bra{#2}}

\newcommand{\hadamard}{\mathrm{H}}

\newcommand{\phase}{\mathrm{S}}

\newcommand{\pieighth}{\mathrm{T}}

\newcommand{\CX}{\mathrm{CX}}

\newcommand{\I}{\mathrm{I}}
\newcommand{\X}{\mathrm{X}}
\newcommand{\Y}{\mathrm{Y}}
\newcommand{\Z}{\mathrm{Z}}
\newcommand{\U}{\blacksquare}
\newcommand{\D}{\heartsuit}

\newcommand{\adj}[1]{#1^{\dagger}}

\newcommand{\pr}[1]{\mathrm{pr}_{#1}}
\newcommand{\semC}[1]{\llbracket #1 \rrbracket_{\mathrm{C}}}
\newcommand{\semCl}[1]{\left\llbracket #1 \right\rrbracket_{\mathrm{C}}}
\newcommand{\semE}[1]{\llbracket #1 \rrbracket_{\mathrm{E}}}
\newcommand{\semEl}[1]{\left\llbracket #1 \right\rrbracket_{\mathrm{E}}}
\newcommand{\QIL}{\mathbf{QIL}}
\newcommand{\Sc}[1][]{\mathcal{S}_{#1}}
\newcommand{\E}[1][]{\mathcal{E}_{#1}}
\newcommand{\F}[1][]{\mathcal{F}_{#1}}
\newcommand{\trans}[2]{{#1}{#2}{\adj{#1}}}
\newcommand{\proj}[2]{{#1}{#2}{#1}}
\newcommand{\appon}[2]{{#1}_{(#2)}}
\newcommand{\Q}{\mathbf{Q}}

\title{Analysis of Quantum Entanglement in Quantum Programs using Stabilizer Formalism}
\author{Kentaro Honda
\institute{Department of Computer Science, the University of Tokyo, Japan}
\email{honda@is.s.u-tokyo.ac.jp}
}
\def\titlerunning{Analysis of Quantum Entanglement in Quantum Programs using Stabilizer Formalism}
\def\authorrunning{K.\ Honda}

\begin{document}
\maketitle

\begin{abstract}
  Quantum entanglement plays an important role in quantum computation and communication.
  It is necessary for many protocols and computations, but causes unexpected disturbance of computational states.
  Hence, static analysis of quantum entanglement in quantum programs is necessary.
Several papers studied the problem.
They decided qubits were entangled if multiple qubits unitary gates are applied to them, and some refined this reasoning using information about the state of each separated qubit.
However, they do not care about the fact that unitary gate undoes entanglement and that measurement may separate multiple qubits.
In this paper, we extend prior work using stabilizer formalism.
It refines reasoning about separability of quantum variables in quantum programs.
\end{abstract}

\section{Introduction}

Quantum entanglement plays an important role in quantum computation and communication.
It allows us to teleport quantum states~\cite{BBCJPW93} and reduce necessary numbers of qubits for communication~\cite{BW92}.
Moreover, it is the essential resource in a one-way quantum computation model~\cite{RB01} and indispensable for outperforming classical computers.
Quantum entanglement also introduces some difficulty in compiling quantum programs.
For example, when a system uses an ancilla, the ancilla is possibly entangled with the computation system and removal of it will disturb the computational state of the system.
Compilers of quantum programs should care about existence of quantum entanglement.
Hence, static analysis of quantum entanglement is necessary.
Several papers studied the problem using types~\cite{P07}, abstract interpretation~\cite{P08b}, and Hoare-like logic~\cite{PZ09}.
The first paper reasoned that two qubits are entangled whenever a two qubits gate is applied to these qubits.
The other papers improved the reasoning by restricting two qubit gates to the controlled-not gate $\CX$ and by memorising information about the basis of separated qubits.
Since $\CX$ does not create entanglement if the control qubit is in $\Z$-basis or the target qubit is in $\X$-basis, we can reason that two qubits are not entangled even after applying $\CX$ to the qubits.
However, these papers do not care about the fact that unitary gate undoes entanglement.
Our motivating example is as follows.
\begin{align*}
\texttt{GHZ} &\equiv \scomp{\texttt{INIT}}{\scomp{\sH{q_0}}{\scomp{\sCX{q_0}{q_1}}{\sCX{q_1}{q_2}}}}\\
\texttt{SEP}_{0} &\equiv \scomp{\texttt{GHZ}}{\scomp{\sCX{q_0}{q_1}}{\sCX{q_0}{q_2}}}
\end{align*}
where \texttt{INIT} changes states of all qubits $q_0, q_1, q_2$ into $\ket{0}$.
\texttt{GHZ} creates GHZ state $\ket{\mathrm{GHZ}} \equiv \frac{1}{\sqrt{2}}(\ket{000}+\ket{111})$, where all qubits are entangled.
$\texttt{SEP}_{0}$ destroys the entanglement without measurement.
Indeed, $(\CX\otimes \I)(\I \otimes \CX)\ket{\mathrm{GHZ}} = \ket{+00}$ and all qubits are separated.
The prior work reasons correctly that entanglement exists after \texttt{GHZ} but incorrectly that entanglement still exists after $\texttt{SEP}_{0}$.
Another example is
\begin{align*}
\texttt{SEP}_{1} &\equiv \scomp{\texttt{GHZ}}{\texttt{meas}(q_0)}\\
\texttt{NSEP} &\equiv \scomp{\texttt{GHZ}}{\scomp{\sH{q_0}}{\texttt{meas}(q_0)}}.
\end{align*}
After executions, $\texttt{SEP}_{1}$ produces all separated qubits but \texttt{NSEP} does one separated and two entangled qubits regardless of the measurement results.
In this paper, we borrow the framework of Perdrix's work~\cite{P08b} and extend it using stabilizer formalism~\cite{AG04,G96,NC00}, which gives a segment of quantum computation that can be classically simulated.
It refines reasoning about separability of quantum variables in quantum programs.

\section{Preliminaries}

\subsection{Stabilizer Formalism}

Stabilizer formalism allows us to express a certain class of states in a compact way.

Let $G_{n}$ be the Pauli group on $n$ qubits.
The stabilizer $S$ of a nontrivial subspace $V_{S}$ of the $2^{n}$-dimensional complex Hilbert space $\mathcal{H}_{2^n}$ is $\{P \in G_{n} \mid \forall \ket{\psi} \in V_{S}.\;\; P \ket{\psi} = \ket{\psi} \}$.
Any stabilizer $S$ is abelian and $-{\I}^{\otimes n} \notin S$.
A subgroup $S$ of $G_n$ is a stabilizer (on $n$ qubits) if it is the stabilizer of some nontrivial subspace of $\mathcal{H}_{2^n}$.
If $\{M_0, \ldots, M_{k-1}\}$ is a set of independent generators of $S$, we use $\langle M_0, \ldots, M_{k-1} \rangle$ to denote $S$.
If $S = \langle M_0, \ldots, M_{k-1} \rangle$, the dimension of $V_{S}$ is $2^{n-k}$.
In particular, if $k=n$, there exists a unique state $\ket{\psi_{S}}$ stabilized by $S$.
We call a state $\ket{\psi}$ is a stabilizer state if $\ket{\psi} = \ket{\psi_{S}}$ for some stabilizer $S$.
$P^{\pm}_{M_i} = \frac{1}{2}(\I^{\otimes n} \pm M_{i})$ is the projection onto eigenspaces corresponding to eigenvalues $\pm 1$.

Stabilizers have matrix expressions.
Let $S = \langle M_0, \ldots, M_{k-1} \rangle$.
Each generator $M_l$ has a form of either $\sigma_{l,0} \otimes \sigma_{l,1} \otimes \cdots \otimes \sigma_{l,n-1}$ or $-\sigma_{l,0} \otimes \sigma_{l,1} \otimes \cdots \otimes \sigma_{l,n-1}$ where $\sigma_{l,m}$ is a Pauli matrix, i.e.\ $\sigma_{l,m} \in \{\I, \X, \Y, \Z\}$.
A stabilizer array~\cite{AP05} is a $k \times (n+1)$ matrix whose $(i,j)$th entry is $\sigma_{i,j}$ for $j < n$ or the sign of $M_{i}$ for $j = n$, and it denotes $S$.
For example, $\langle -\Z\Z, \X\X \rangle = \{\I, \X\X, \Y\Y, -\Z\Z \}$ stabilizes $\frac{1}{\sqrt{2}}(\ket{01}+\ket{10})$.
$\left[\begin{array}{*3{c}} \Z & \Z & -\\ \X & \X & + \end{array} \right]$ is a stabilizer array of the stabilizer.
We identify the $i$th row of a stabilizer array and the generator $M_i$.
Obviously, both permutation of rows and multiplication of the $i$th row and the $j$th row do not change the stabilizer provided $i \neq j$ where ``multiplication of the $i$th row and the $j$th row'' is replacement of the $i$th row with the product of the $i$th row and the $j$th row.
Stabilizer arrays are compact but have sufficient information to their stabilizers.
We use stabilizer arrays to operate stabilizers.

Let $S = \langle M_0, \ldots, M_{k-1} \rangle$ and $T = \langle N_0, \ldots, N_{l-1} \rangle$ be stabilizers on $k$ and $l$ qubits.
Their tensor product $S \otimes T$ is the stabilizer $\langle M_0 \otimes \I^{\otimes l}, \ldots, M_{k-1} \otimes \I^{\otimes l}, \I^{\otimes k} \otimes N_0, \ldots, \I^{\otimes k} \otimes N_{l-1} \rangle$ on $k+l$ qubits.
In stabilizer array expression, the tensor product is the direct sum of two matrices.

When $S = \langle M_0, \ldots, M_{n-1} \rangle$ is the stabilizer of $V_{S}$, $\trans{U}{S} = \langle \trans{U}{M_0}, \ldots, \trans{U}{M_{n-1}} \rangle$ ``stabilizes'' $UV_{S}$ for any unitary gate $U$.
However, some $\trans{U}{M_i}$ may exceed $G_{n}$ and hence may not be a stabilizer.
A Clifford gate is a unitary gate that sends any stabilizer to a stabilizer.
Any Clifford gate can be composed of the controlled-X gate $\CX$, the Hadamard gate $\hadamard$, and the phase gate $\phase$.
A well-known non-Clifford gate is the $\frac{\pi}{8}$-gate $\pieighth$.
Indeed, $\pieighth\X\adj{\pieighth} = \frac{1}{\sqrt{2}}(\X+\Y)$ and $\pieighth \ket{+} = \frac{1}{\sqrt{2}}(\ket{0}+e^{\frac{\pi}{4}}\ket{1})$ is not a stabilizer state.

Let $\langle M_0, \ldots, M_{n-1} \rangle$ be a stabilizer on $n$ qubits.
If any $M_i$ commutes with $\appon{Z}{j} \equiv \I^{\otimes j} \otimes \Z \otimes \I^{\otimes n-j-1}$, i.e.\ the $j$th column of a stabilizer array consists of $\I$ or $\Z$, then the measurement result of the $j$th qubit is deterministic and does not change the state.
If not, the measurement result is probabilistic.
Through multiplication of rows, we can take a unique generator $M_{i}$ that does not commute with $\appon{Z}{j}$.
The stabilizer of the post-measurement state is $\langle M_0, \ldots, M_{i-1}, \pm \appon{Z}{j}, M_{i+1}, \ldots, M_{n-1}\rangle$ if the measurement result is $\pm 1$, respectively.

\subsection{Quantum Imperative Language}

Following prior work~\cite{P08b}, we use Quantum Imperative Language~(QIL) as a target language.
Fix the set $\Q$ of quantum variables $\{\texttt{q}_\texttt{0}, \ldots, \texttt{q}_{N-1}\}$.
We assume $\Q$ is finite and often identify a quantum variable and its index.
The syntax of QIL~\cite{P08a} is the following.
\begin{align*}
  C, C' & \Coloneqq
  \sskip
  \mid \scomp{C}{C'}
  \mid \sX{i}
  \mid \sY{i}
  \mid \sZ{i}
  \mid \sH{i}
  \mid \sS{i}
  \mid \sT{i}
  \mid \sCX{i}{j}\\
  &
  \mid \sif{i}{C}{C'}
  \mid \swhile{i}{C}
\end{align*}
where $i \neq j$.
$\QIL$ is the set of QIL programs.
The concrete denotational semantics of QIL is a superoperator $\sem{\cdot} \colon \QIL \rightarrow D_{2^{N}} \rightarrow D_{2^{N}}$ where $D_{n}$ is the set of $n$-dimensional partial density matrices, which is a CPO~\cite{S04}.
\begin{align*}
  \sem{\sskip}(\rho) & = \rho\\
  \sem{\scomp{C}{C'}}(\rho) & = \sem{C'}(\sem{C} (\rho))\\
  \sem{U(i)}(\rho) &= \trans{\appon{U}{i}}{\rho}\\
  \sem{\sCX{i}{j}}(\rho) &= \trans{\appon{\CX}{i,j}}{\rho}\\
  \sem{\sif{i}{C}{C'}}(\rho) &=
  \sem{C}(\proj{\appon{\ketbra{0}{0}}{i}}{\rho})
  + \sem{C'}(\proj{\appon{\ketbra{1}{1}}{i}}{\rho})\\
  \sem{\swhile{i}{C}}(\rho) &=
  \sum_{n \in \N}
  \proj{\appon{\ketbra{1}{1}}{i}}{{f}^{n}(\rho)}
\end{align*}
where $U \in \{ \tX, \tY, \tZ, \tH, \tS, \tT\}$, $f(\rho) = \sem{C}(\proj{\appon{\ketbra{0}{0}}{i}}{\rho})$.

QIL has a control structure and hence we can change any state of a quantum variable into a constant.
\begin{align*}
  \texttt{INIT}_{i} & \equiv \sif{i}{\sskip}{\sX{i}}\\
  \texttt{INIT} & \equiv \scomp{\texttt{INIT}_{0}}{\scomp{\texttt{INIT}_{1}}{\scomp{\cdots}{\texttt{INIT}_{N-1}}}}
\end{align*}
Indeed, $\sem{\texttt{GHZ}}(\rho) = \ketbra{\mathrm{GHZ}}{\mathrm{GHZ}}$ and $\sem{\texttt{SEP}_{0}}(\rho) = \ketbra{+00}{+00}$.

In the work~\cite{P08b}, an abstract domain $A^{\Q}$ to analyse entanglement was introduced.
An element of the domain is a pair $(b, \pi)$ of a partition $\pi$ of $\Q$ and a function $b \colon \Q \to \{\I, \X, \Z, \top\}$.
$\pi$ denotes that the quantum state $\rho$ is $\pi$-separable:
\begin{equation*}
  \rho = \sum_{k} p_k \bigotimes_{A_j \in \pi} \rho^{k,j}
\end{equation*}
where $\rho^{k,j}$ is a quantum state of $A_j$.
Moreover, if the $i$th qubit is separated from the others, $b(i)$ shows which basis it is.
For example, if $b(i) = \Z$, the quantum state $\rho$ is:
\begin{equation*}
  \rho = p_0 \ketbra{0}{0}\otimes\rho_0 + p_1 \ketbra{1}{1}\otimes\rho_1
\end{equation*}
for some $p_0, p_1, \rho_0, \rho_1$.
It implies that the $i$th qubit will be still separated even after executing $\sCX{i}{j}$.

\section{Abstract domain on stabilizers}

Although $A^{\Q}$ gives us some information about separability of a quantum state, it contains no information about entanglement except that qubits are entangled.
In order to analyse more, we will refine the abstract domain $A^{\Q}$ using stabilizer formalism.
We follow the idea of $A^{\Q}$, where $\Z$ and $\X$ denote that a state can be separated through $\ket{0}, \ket{1}$ and $\ket{+}, \ket{-}$, respectively.
We suppose that a stabilizer $S = \langle M_0, \ldots, M_{n-1}\rangle$ on $n$ qubits represents not only the stabilized state $\ket{\psi_{S}}$ but also the eigenstates of it, i.e.\ $\{\ket{\psi} \mid \forall M_i\;\; M_i\ket{\psi} = \ket{\psi} \;\mbox{or}\; M_i\ket{\psi} = -\ket{\psi}\}$.
We reuse $\ket{\psi_{S}}$ to denote an eigenstate.
The sign of each generator has no longer any meaning.
From now on, we assume any generator has the plus sign and we ignore the last column of any stabilizer array.

Our idea of using stabilizers, of course, has a problem about non-Clifford gates.
Since QIL has the $\frac{\pi}{8}$-gate $\tT$, even if we start an execution of a QIL program from a stabilizer state, we may not get a stabilizer state.
We prepare a symbol $\U$ that denotes a non-stabilizer.

Now, we introduce our abstract domain $C^{\Q}$, which is composed of assignments of stabilizers to each segment of partitions of $\Q$.
When $\sT{i}$ appears, we forget about a stabilizer that expresses the current state of the segment containing the $i$th qubit, and keep just the symbol $\U$.
Hence, when we can divide a stabilizer into the tensor product of multiple stabilizers, it is good to separate them.
In particular, if a stabilizer on multiple qubits contains $\appon{\X}{i}$, $\appon{\Y}{i}$, or $\appon{\Z}{i}$, then the $i$th qubit can be separated from the others.
Naive algorithms on a stabilizer array allow us to compute whether $\appon{\X}{i}$ belongs to a given stabilizer in $O(N)$ time and to divide a stabilizer into two stabilizers in $O(N^{2})$ time.

\begin{dfn}
  Let $\Sc[k]$ be the set of stabilizers on $k \geq 2$ qubits that are generated by $k$ independent generators and contain none of $\appon{\X}{i}$, $\appon{\Y}{i}$, and $\appon{\Z}{i}$.
  $\Sc[1] = \{\I, \langle \X \rangle, \langle \Y \rangle, \langle \Z \rangle\}$.
  We add the non-stabilizer $\U$ to all $\Sc[k]$.
  Define $\Sc = \bigcup_{k \leq N} \Sc[k]$.
  We call $\alpha \subset 2^{\Q} \times \Sc$ a \textit{(stabilizer) assignment} if $\pr{0} \alpha$ is a partition of $\Q$ and for any $(A, S) \in \alpha$, $S \in \Sc[|A|]$.
  Here, $\pr{i}$ is the $i$th projection.
  The set of stabilizer assignments is $C^{\Q}$.
\end{dfn}
\begin{ntt}
  Let $\alpha$ be an assignment.
  We sometimes regard $\alpha$ as a function from $\Q$ to $2^{\Q} \times \Sc$ such that $\alpha(i) = (A, S)$ where $i \in A$.
  We define $\alpha_0 = \pr{0} \circ \alpha, \alpha_1 = \pr{1} \circ \alpha$.
  Hence, $\alpha_0(i) \in 2^{\Q}$ and $\alpha_1(i) \in \Sc$.
  We also regard a partition of $\Q$ as a function from $\Q$ to $2^{\Q}$.
  $\alpha[(A, S)/i]$ is a new assignment $(\alpha \backslash \alpha(i)) \cup \{(A, S)\}$.
  We extend the notation into $\alpha[\{(A_0, S_0), \ldots, (A_{k-1}, S_{k-1})\}/i]$ in a natural manner.
  $\alpha[S/i]$ means $(\alpha \backslash \alpha(i)) \cup \{(\alpha_0(i), S)\}$.
  $\alpha[S/i, j] = (\alpha \backslash (\alpha(i) \cup \alpha(j))) \cup \{(\alpha_0(i)\cup\alpha_0(j), S)\}$.
\end{ntt}

\begin{dfn}
  Let $\rho$ be a quantum state and $\alpha$ be an assignment.
  We write $\alpha \vDash \rho$ if
  \begin{equation*}
    \rho = \sum_{k} p_k \bigotimes_{(A, S) \in \alpha} \rho^{k,(A,S)}
  \end{equation*}
  with some probability $p_k$ and some state $\rho^{k,(A,S)}$ on $A$ qubits where $\rho^{k,(A,S)}$ has a form of $\frac{1}{2}\I$ if $S = \I$ and $\ketbra{\psi_{S}}{\psi_{S}}$ if $S$ is another stabilizer.
\end{dfn}
Although an assignment tells how to separate a quantum state, it is just an overapproximation.
Even if a stabilizer is assigned to two qubits, it does not mean the qubits are entangled.
Indeed, although $\frac{1}{4}(\I \otimes \I)$ is a separable state, $\{(\{0,1\}, \langle \X\X, \Z\Z \rangle)\} \vDash \frac{1}{4} (\I \otimes \I)$.

Each assignment contains information about entanglement of a quantum state.
Intuitively, an assignment $\alpha$ has more information than another assignment $\beta$ if $\beta \vDash \rho$ whenever $\alpha \vDash \rho$.
It gives $C^{\Q}$ a lattice structure:
For $S, S' \in \Sc$, we write $S \leq_{s} S'$ if $S = \I$, $S' = \U$, or $S = S'$. Obviously, $\leq_{s}$ is an order.
Let $\leq_{\pi}$ be an order of partitions: $\pi \leq_{\pi} \pi'$ if for any $A' \in \pi'$, there exist $A_0, \ldots A_{k-1} \in \pi$ such that $A' = \bigcup_{i \in \{0, \ldots, k-1\}} A_i$.
Moreover, we write $\alpha \leq_{c} \beta$ if $\alpha_{0} \leq_{\pi} \beta_{0}$ and for each $i \in \Q$, 
$\bigodot_{j \in \beta_{0}(i)} \alpha(j) \leq_{s} \beta_{1}(i)$ where
\begin{equation*}
  \bigodot_{j \in J} (A_j, S_j) = \left\{
  \begin{array}{cl}
    S_j & (\mbox{all $A_j$ are the same})\\
    \I & (\mbox{all $S_j$ are $\I$})\\
    \U & (\mbox{otherwise})
  \end{array}
  \right..
\end{equation*}
The relation $\leq_{c}$ makes $C^{\Q}$ a CPO.\@

\begin{prp}
  $C^{\Q}$ is a finite lattice and hence a CPO.\@
\end{prp}
\begin{proof}
  It is easy to see $\leq_{c}$ is an order.
  The maximum assignment is $\{(\Q, \U)\}$ and the minimum is $\{(\{i\}, \I) \mid i \in \Q\}$.
  Let $\alpha, \beta$ be assignments.
  Take the join $\pi$ of $\alpha_{0}$ and $\beta_{0}$ with respect to $\leq_{\pi}$.
  Let $A \in \pi$.
  Define $S$ as the join of $\bigodot_{j \in A} \alpha_{1}(j)$ and $\bigodot_{j \in A} \beta_{1}(j)$.
  The set of these pairs $(A, S)$ is the join of $\alpha$ and $\beta$.
  The meet of $\alpha$ and $\beta$ can be constructed similarly.
\end{proof}

We define an abstract semantics $\semC{\cdot} \colon \QIL \rightarrow C^{\Q} \rightarrow C^{\Q}$ inductively.
For simplicity, we define $\trans{U}{\U} = \U$ for any unitary $U$ and $\U \otimes S = \U = S \otimes \U$ for any $S \in \Sc$, and assume that conditions are exclusive and an upper condition has priority.
\begin{align*}
  \semC{\sskip}(\alpha) &= \alpha\\
  \semC{\scomp{C}{C'}}(\alpha) &= \semC{C'}(\semC{C}(\alpha))\\
  \semC{U(i)}(\alpha) &= \alpha[\trans{\appon{U}{i}}{\alpha_1(i)}/i]\\
  \semC{\sT{i}}(\alpha) &=
  \left\{
  \begin{array}{ll}
    \alpha & (\alpha_{1}(i) \;\mbox{and $\appon{\Z}{i}$ commute})\\
    \alpha[\U/i] & (\mbox{otherwise})
  \end{array}
  \right.\\
  \semC{\sCX{i}{j}}(\alpha) &=
  \left\{
  \begin{array}{ll}
    \alpha[update(\{i,j\}, \alpha_{0}(i), \trans{\appon{\CX}{i,j}}{\alpha_1(i)})/i] & (\alpha_{0}(i) = \alpha_{0}(j))\\
     & (\alpha_{1}(i) = \langle\Z\rangle,\\
    \alpha & \;\;\alpha_{1}(j) = \langle\X\rangle, \mbox{ or}\\
    & \;\;\alpha_{1}(i) = \alpha_{1}(j) = \I)\\
    \alpha[\langle Z \rangle/i] & (\alpha_{1}(i) = \I)\\
    \alpha[\langle X \rangle/j] & (\alpha_{1}(j) = \I)\\
    \alpha[\trans{\appon{\CX}{i,j}}{(\alpha_{1}(i) \otimes \alpha_{1}(j))}/i,j] & (\mbox{otherwise})
  \end{array}
  \right.\\
  \semCl{
    \begin{array}{rlc}
      \tif & i &\\
      & \tthen & C\\
      & \telse & C'\\
      \tfi & &
    \end{array}
  }(\alpha) &=
  \semC{C}(meas(i, \alpha)) \vee \semC{C'}(meas(i, \alpha))\\
  \semC{\swhile{i}{C}}(\alpha) &=
  \bigvee_{n \in \N} meas_{i} ((\semC{C} \circ meas_{i})^{n} (\alpha))
\end{align*}
where $U \in \{\tX, \tY, \tZ, \tH, \tS\}$.
$update$ makes a ``pseudo-''assignment to satisfy the condition that each stabilizer contains none of $\appon{\X}{i}$, $\appon{\Y}{i}$, and $\appon{\Z}{i}$.
The first argument is possibly-unentangled qubits.
\begin{align*}
  update(J, A, \U) &= \{(A, \U)\}\\
  update(\emptyset, A, S) &= \{(A, S)\}\\
  update(\{i\} \cup J, A, S) &=
  \left\{
  \begin{array}{ll}
    \{(\{i\}, S')\} \cup update(J, A \backslash \{i\}, S'') & (S = S' \otimes S'' \;\mbox{such that}\; S' \in \Sc[1]\;\mbox{and $S'$ has}\\
    & \;\mbox{non-identity entry only in the $i$th column})\\
    update(J, A, S) & (\mbox{otherwise})\\
  \end{array}
  \right.
\end{align*}
$meas$ means measurement and $meas_{i}(\alpha) = meas(i, \alpha)$.
After measurement, the measured qubit is always separated.
\begin{align*}
  meas(i, \alpha) &=
  \left\{
  \begin{array}{ll}
    \alpha[\langle \Z \rangle/i] & (|\alpha_{0}(i)| = 1)\\
    \alpha[\{(\{i\}, \langle \Z \rangle), (\alpha_{0}(i) \backslash \{i\}, \U)\}/i] & (\alpha_{1}(i) = \U)\\
    \alpha[update(\alpha_{0}(i), \alpha_{0}(i), meas_{st}(i, \alpha_{1}(i)))/i]
    & (\mbox{otherwise})
  \end{array}
  \right.
\end{align*}
where $meas_{st}$ is the measurement process of the $i$th qubit in stabilizer formalism.

\begin{exm}
  $\semC{\texttt{GHZ}}(\alpha) = \{(\{0,1,2\},\langle \X\X\X, \Z\Z\I, \I\Z\Z\rangle)\}$, $\semC{\texttt{SEP}_{0}}(\alpha) = \{(\{0\},\langle \X\rangle), (\{1\},\langle \Z\rangle), (\{2\},$ $\langle \Z \rangle)\}$, $\semC{\texttt{SEP}_{1}}(\alpha) = \{(\{0\},\langle \Z\rangle), (\{1\},\langle \Z\rangle), (\{2\},\langle \Z \rangle)\}$, and $\semC{\texttt{NSEP}}(\alpha) = \{(\{0\},\langle \Z\rangle), (\{1,2\},\langle \X\X,$ $\Z\Z\rangle)\}$ where $\texttt{meas}(i) \equiv \sif{i}{\sskip}{\sskip}$.
\end{exm}

\begin{exm}
  Take a QIL program $\texttt{exm}_{0} = \scomp{\sT{0}}{\sif{1}{\sskip}{\sCX{2}{3}}}$.
  Let $\alpha_{\texttt{exm}_{0}} = \{(\{0,1\}, \langle \Z\Z, \X\X \rangle ), (\{2, 3\}, \langle \Z\Z, \X\X \rangle )\}$, $\ket{B_{00}}$ be a Bell state $\frac{1}{\sqrt{2}}(\ket{00}+\ket{11})$, and $\rho_{\texttt{exm}_{0}} = \ketbra{B_{00}}{B_{00}} \otimes \ketbra{B_{00}}{B_{00}}$.
  Since $\ket{B_{00}}$ is stabilized by $\langle \Z\Z, \X\X \rangle$, $\alpha_{\texttt{exm}_{0}} \vDash \rho_{\texttt{exm}_{0}}$.
  \begin{align*}
  \left[
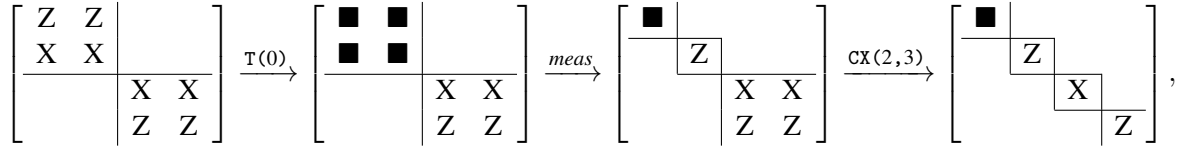
    \begin{array}{*4{c}}
      \Z & \multicolumn{1}{c|}{\Z} &&\\
      \X & \multicolumn{1}{c|}{\X} &&\\
      \cline{1-4}
      && \multicolumn{1}{|c}{\X} & \X \\
      && \multicolumn{1}{|c}{\Z} & \Z \\
    \end{array}
    \right]
  \xrightarrow{\sT{0}}
    \left[
    \begin{array}{*4{c}}
      \U & \multicolumn{1}{c|}{\U} &\\
      \U & \multicolumn{1}{c|}{\U} &\\
      \cline{1-4}
      && \multicolumn{1}{|c}{\X} & \X \\
      && \multicolumn{1}{|c}{\Z} & \Z \\
    \end{array}
    \right]
  \xrightarrow{meas}
    \left[
    \begin{array}{*4{c}}
      \multicolumn{1}{c|}{\U} &  &\\
      \cline{1-2}
       & \multicolumn{1}{|c|}{\Z} &\\
      \cline{2-4}
      && \multicolumn{1}{|c}{\X} & \X \\
      && \multicolumn{1}{|c}{\Z} & \Z \\
    \end{array}
    \right]
  \xrightarrow{\sCX{2}{3}}
  \left[
    \begin{array}{*4{c}}
      \multicolumn{1}{c|}{\U} &  &\\
      \cline{1-2}
      & \multicolumn{1}{|c|}{\Z} &\\
      \cline{2-3}
      && \multicolumn{1}{|c|}{\X} &  \\
      \cline{3-4}
      &&& \multicolumn{1}{|c}{\Z}\\
    \end{array}
    \right],
\end{align*}
  so $\semC{\texttt{exm}_{0}}(\alpha_{\texttt{exm}_0})$ is $\{(\{0\}, \U), (\{1\}, \Z), (\{2, 3\}, \U)\}$, which satisfies $\semC{\texttt{exm}_{0}}(\alpha_{\texttt{exm}_0}) \vDash \sem{\texttt{exm}_{0}}(\rho_{\texttt{exm}_0}) = \frac{1}{2}\ketbra{0}{0}\otimes\ketbra{0}{0}\otimes\ketbra{B_{00}}{B_{00}}+\frac{1}{2}\ketbra{1}{1}\otimes\ketbra{1}{1}\otimes\ketbra{{+}0}{{+}0}$.
  Note the join of $\{(\{2\}, \langle \X\rangle), (\{3\}, \langle \Z\rangle)\}$ and $\{(\{2,3\}, \langle \X\X, \Z\Z\rangle)\}$ is $\{(\{2, 3\}, \U)\}$.
\end{exm}

In the above example, we can see that $\tCX$ undoes quantum entanglement between the second and third qubits.
It enables us to analyse entanglement in a QIL program more deeply than the prior work.
Of course, in order to use $\semC{\cdot}$ for analysis, the abstract semantics should be sound for the concrete semantics.
Although $\semC{\cdot}$ is not generally monotone, it is sound as the abstract semantics in the paper~\cite{P08b} is.
A counterexample of $\semC{\cdot}$ being monotone is
$\alpha = \{(\{0\}, \I), (\{1\}, \langle \Z \rangle)\}$,
$\beta = \{(\{0\}, \langle \X \rangle), (\{1\}, \langle \Z \rangle)\}$, and
$C = \scomp{\scomp{\sCX{0}{1}}{\sS{1}}}{\scomp{\sH{0}}{\scomp{\sCX{0}{1}}{\sT{1}}}}$.
$\alpha <_{c} \beta$ but $\semC{C}(\alpha) >_{c} \semC{C}(\beta)$ because $\semC{C}(\alpha) = \{(\{0,1\}, \U)\}$ and $\semC{C}(\beta) = \{(\{0\}, \langle \Y \rangle), (\{1\}, \langle \U \rangle)\}$.

\begin{prp}
  For any assignment $\alpha, \beta$, and QIL program $C$, if $\alpha \leq_{c} \beta$ and $\alpha_{1}(i) \neq \I$ for any $i \in \Q$, then $\semC{C}(\alpha) \leq_{c} \semC{C}(\beta)$.
\end{prp}
\begin{proof}
  By induction on the structure of $C$.
\end{proof}

\begin{thm}
  For any state $\rho$, assignment $\alpha$, and QIL program $C$, $\alpha \vDash \rho$ implies $\semC{C}(\alpha) \vDash \sem{C}(\rho)$.
\end{thm}
\begin{proof}
  By induction on the structure of $C$.
  For $\sskip$, $\scomp{C}{C'}$, $U(i)$, and $\sT{i}$, it is easy.
  For $\sCX{i}{j}$, there are several cases.
  But, in any case, it is straightforward that the statement holds by the definition of $\alpha \vDash \rho$ and computation in stabilizer formalism.
  Note that $\alpha \vee \beta \vDash \rho + \sigma$ whenever $\alpha \vDash \rho$ and $\beta \vDash \sigma$.
  The statement holds for $\sif{i}{C}{C'}$ because of the above fact, $meas(\alpha) \vDash \proj{\ketbra{0}{0}}{\rho}$, and $meas(\alpha) \vDash \proj{\ketbra{1}{1}}{\rho}$.
  Finally, we show for $\swhile{i}{C}$.
  Because of $meas(\alpha) \vDash \proj{\ketbra{0}{0}}{\rho}$ and the induction hypothesis, $\bigvee_{n \leq M} meas ( (\semC{C} \circ meas)^{n} (\alpha)) \vDash \sum_{n \leq M} \proj{\appon{\ketbra{1}{1}}{i}}{{f}^{n}(\rho)}$.
  Since $C^{\Q}$ is finite, $\semC{\swhile{i}{C}}(\alpha) \vDash \sum_{n \leq M} \allowbreak \proj{\appon{\ketbra{1}{1}}{i}}{{f}^{n}(\rho)}$ for sufficiently large $M$.
  Thus, the statement holds.
\end{proof}

\section{Abstract domain on extended stabilizers}

In the previous section, we use stabilizers and the symbol $\U$ that represents a non-stabilizer.
The symbol $\U$ contains no information.
It just states that the state of the associated qubits is unknown.
The abstract semantics $\semC{\cdot}$ introduces the symbol when it faces the non-Clifford gate $\tT$ because the post-execution state is a non-stabilizer state.
Can not we really extract meaningful information from the post-execution state?
Let us take the following QIL program.
\begin{equation*}
  \texttt{exm}_{1} \equiv \scomp{\texttt{GHZ}}{\scomp{\sT{1}}{\texttt{meas}(0)}}
\end{equation*}
The abstract semantics
\begin{equation*}
  \left[
    \begin{array}{*3{c}}
      \X & \X & \X\\
      \Z & \Z & \I\\
      \Z & \I & \Z
    \end{array}
    \right]
  \xrightarrow{\sT{1}}
    \left[
    \begin{array}{*3{c}}
      \U & \U & \U\\
      \U & \U & \U\\
      \U & \U & \U
    \end{array}
    \right]
    \xrightarrow{\texttt{meas}(0)}
  \left[
    \begin{array}{*4{c}}
      \multicolumn{1}{c|}{\Z} &  &\\
      \cline{1-3}
      & \multicolumn{1}{|c}{\U} & \U \\
      & \multicolumn{1}{|c}{\U} & \U \\
    \end{array}
    \right]
\end{equation*}
tells us that the first qubit is separated but the second and the third qubits may be entangled last.
Now, let us try not to fill the matrix with $\U$ when $\tT$ appears but to memorise the applied gates.
Recall that $US\adj{U}$ ``stabilizes'' $UV_{S}$ if $S$ is the stabilizer of $V_{S}$.

\begin{equation*}
  \left[
    \begin{array}{*3{c}}
      \X & \X & \X\\
      \Z & \Z & \I\\
      \Z & \I & \Z
    \end{array}
    \right]
  \xrightarrow{\sT{1}}
  \left[
    \begin{array}{*3{c}}
      \X & \pieighth\X\adj{\pieighth} & \X\\
      \Z & \Z & \I\\
      \Z & \I & \Z
    \end{array}
    \right]
  \xrightarrow{\texttt{meas}(0)}
  \left[
    \begin{array}{*4{c}}
      \multicolumn{1}{c|}{\Z} &  &\\
      \cline{1-2}
      & \multicolumn{1}{|c|}{\Z} & \\
      \cline{2-3}
      && \multicolumn{1}{|c}{\Z}
    \end{array}
    \right]
\end{equation*}
It means all qubits are separated.
Indeed, $\sem{\texttt{exm}_{1}}(\rho) = \frac{1}{2}(\ketbra{000}{000}+\ketbra{111}{111})$.
The example shows that the effect of $\tT$ may be bounded locally and will be removed later.
We introduce a new symbol $\D$, which means a unitary matrix that may not be a Pauli matrix or their tensor product.
Note that $\D$ means not only a single qubit unitary matrix but also an $n$ qubit unitary matrix.
Using the symbol $\D$, we will extend our abstract domain $C^{\Q}$ to a new domain $E^{\Q}$.
Before doing it, we extend stabilizers so that they allow us to put $\D$ on them.

\begin{dfn}
  Let $k$ be a positive integer and $A$ be a $k \times k$ matrix whose entries are in $\{\I, \X, \Y, \Z, \D\}$.
  We now identify two matrices $A$ and $B$ if one can be converted into the other via permutations and multiplications of rows.
  Here, $\D$ behaves as an absorbing element.
  We call a row containing the symbol $\D$ and a row containing no $\D$ a $\D$-row and an L-row, respectively.
  We always require any L-rows commute and these rows are independent.
  Moreover, we require that for any $\D$-row $R_i$ and row $M_j$, by substituting $\I$, $\X$, $\Y$, or $\Z$ for each $\D$ in $R_i$ and $M_j$, the result rows can commute.
  For example, the matrix consisting of two rows $\D\X$ and $\I\Z$ is excluded, but the matrix consisting of $\D\X$ and $\X\Z$ is right because substitution of $\Z$ for $\D$ makes these rows commute.
  Finally, we exclude some matrices.
  Let $k \geq 2$.
  If a matrix has a row $\I\I\cdots\I$, $\appon{\X}{j}$, $\appon{\Y}{j}$, $\appon{\Z}{j}$, or $\appon{\D}{j}$, then it is excluded.
  We name the set of those matrices $\E[k]$.
  Then, $\U$ is added into all $\E[k]$.
  $\E$ is the union of these $\E[k]$s.
\end{dfn}
\begin{exm}
  \begin{equation*}
  \left[
    \begin{array}{c}
      \I
    \end{array}
    \right],
  \left[
    \begin{array}{c}
      \D
    \end{array}
    \right],
  \left[
    \begin{array}{*2{c}}
      \X & \D\\
      \D & \X
    \end{array}
    \right],
  \left[
    \begin{array}{*3{c}}
      \D & \X & \Y\\
      \Z & \D & \D\\
      \X & \Y & \Z
    \end{array}
    \right]
  \in \E, \;\;\;\;
  \left[
    \begin{array}{ccc}
      \D & \Z & \Y\\
      \I & \X & \Y\\
      \Z & \D & \D
    \end{array}
    \right],
  \left[
    \begin{array}{cc}
      \I & \Z\\
      \X & \D
    \end{array}
    \right],
  \left[
    \begin{array}{cc}
      \D & \Y\\
      \I & \X
    \end{array}
    \right]
  \notin \E
  \end{equation*}
  The third matrix is an abstraction of matrices such as
  \begin{equation*}
  \left[
    \begin{array}{cc}
      \X & \I\\
      \I & \X
    \end{array}
    \right],
  \left[
    \begin{array}{cc}
      \X & \Z\\
      \Z & \X
    \end{array}
    \right],
  \left[
    \begin{array}{cc}
      \X & \trans{\hadamard}{\trans{\pieighth}{\X}}\\
      \Z & \X
    \end{array}
    \right].
  \end{equation*}
\end{exm}

Recall that $\Sc$ has the order $\leq_{s}$.
Regardless of the addition of $\D$, the same definition seems to give an order of $\E$:
$E \in \E$ is lower than or equal to $E' \in \E$ if $E = \I$, $E' = \U$, or $E = E'$.
However, it does not answer our purpose.
Recall the join operator corresponds with the summation of density matrices.
For example, 
$\left[\begin{array}{cc} \X & \D\\ \D & \X \end{array}\right]$
may represent 
$\left[\begin{array}{cc} \X & \I\\ \I & \X \end{array}\right]$
or
$\left[\begin{array}{cc} \X & \Z\\ \Z & \X \end{array}\right]$.
However, the summation of stabilized states by them does not always have the form of $\left[\begin{array}{cc} \X & \D\\ \D & \X \end{array}\right]$.
The example shows the join of $\left[\begin{array}{cc} \X & \D\\ \D & \X \end{array}\right]$ and $\left[\begin{array}{cc} \X & \D\\ \D & \X \end{array}\right]$ should not be itself, so the ``order'' is not reflexive.

In order to obtain a join operator, we remove rows that contain $\D$.
We give up keeping information about unitary matrices when we take a join.
\begin{dfn}
  Take $A \in \E[k]$.
  Remove all $\D$-rows.
  If all rows are $\D$-rows, we obtain $\U$.
  We call this procedure \textit{normalisation} and these matrices \textit{normal forms}.
  The set of normal forms is $\F[k]$ and the union of them is $\F$.
  We redefine $\E[k]$ so that it includes any element of $\E[k]$ even if some rows are removed.
  $\E$ is the union of them.
  Note $\F[k] \subset \E[k]$ and thus $\F \subset \E$.
\end{dfn}
\begin{ntt}
   For each $E \in \E$, $E_{nl}$ is the normal form of $E$.
\end{ntt}
\begin{exm}
  \begin{equation*}
  \left[
    \begin{array}{c}
      \I
    \end{array}
    \right],
  \left[
    \begin{array}{c}
      \U
    \end{array}
    \right],
  \left[
    \begin{array}{ccc}
      \X & \Y & \Z
    \end{array}
    \right]
  \in \F
  \end{equation*}
\end{exm}

$\F$ has an order $\leq_{f}$: $F \leq_{f} F'$ if $F = \I$, $F' = \U$, or $F = F'$.
Obviously, $\F$ has the maximum, the minimum, and the join and the meet of any two elements.
We can take an approximation of a join operator of $\E$ via the subset $\F$.

Now, we define our second abstract domain $E^{\Q}$.
\begin{dfn}
  We call $\gamma \subset 2^{\Q} \times \mathcal{E}$ an \textit{extended (stabilizer) assignment} if $\pr{0} \gamma$ is a partition of $\Q$ and for any $(A, E) \in \gamma$, $E \in \mathcal{E}_{|A|}$.
  The set of extended assignments is $E^{\Q}$.
  For each extended assignment $\gamma$, an extend assignment $\{(A, E_{nl}) \mid (A, E) \in \gamma\}$ is the \textit{normal form} of $\gamma$.
  $F^{\Q}$ is the set of normal forms of extended assignments.
\end{dfn}
\begin{ntt}
 For extended assignments, we use the same notation as for assignments.
\end{ntt}
\begin{dfn}
  Let $\rho$ be a quantum state and $\gamma$ be an extended assignment.
  We write $\gamma \vDash \rho$ if
  $\rho$ is $\pr{0}\gamma$-separable, for any L-row $L_k$ of any $\gamma_{1}(i)$ that is not $\I$ or $\U$, $P^{+}_{L_k}\rho P^{-}_{L_k} = 0$, and for any $i$ such that $\gamma_{1}(i) = \I$, $\rho = \frac{1}{2} \appon{\I}{i} \otimes \rho'$ with some state $\rho'$ of the $\Q \backslash \{i\}$ qubits.
  Recall $P^{\pm}_{L_k} = \frac{1}{2}(\I^{\otimes n} \pm L_{i})$.
\end{dfn}

The same construction as $C^{\Q}$ makes $F^{\Q}$ a CPO.\@
Although $E^{\Q}$ does not have joins, we can define an approximate join operator $\uplus$ on $E^{\Q}$ through $F^{\Q}$: for each $\gamma, \delta \in E^{\Q}$, $\gamma \uplus \delta$ is the join of the normal forms of $\gamma$ and $\delta$.
Note that the approximate join $\uplus$ of two elements can be computed efficiently.
Now, we define our second abstract semantics $\semE{\cdot} \colon \QIL \rightarrow E^{\Q} \rightarrow E^{\Q}$.
Since $\D$ loses some information, we have to avoid introducing $\D$ if possible.
For simplicity, we define $\trans{U}{\U} = \U$ for any $U$, $\U \otimes E = \U = E \otimes \U$ for any $E \in \E$, $\trans{U}{\D} = \D$ for any $1$ qubit unitary $U$, and $\trans{\CX}{(\D U)} = \trans{\CX}{(U\D)} = \D\D$ for any $U$.
Moreover, we assume that conditions are exclusive and an upper condition has priority.
\begin{align*}
  \semE{\sskip}(\gamma) &= \gamma\\
  \semE{\scomp{C}{C'}}(\gamma) &= \semE{C'}(\semE{C}(\gamma))\\
  \semE{U(i)}(\gamma) &= \gamma[\trans{\appon{U}{i}}{\gamma_1(i)}/i]\\
  \semE{\sT{i}}(\gamma) &=
  \left\{
  \begin{array}{ll}
    \gamma & (\gamma_{1}(i) \;\mbox{and $\appon{\Z}{i}$ commute})\\
    \gamma[add_{\D}(i,\gamma_{1}(i))/i] & (\mbox{otherwise})
  \end{array}
  \right.\\
  \semE{\sCX{i}{j}}(\gamma) &=
  \left\{
  \begin{array}{ll}
    \gamma[update_{E}(\{i,j\}, \gamma_{0}(i), \trans{\appon{\CX}{i,j}}{\gamma_1(i)})/i] & (\gamma_{0}(i) = \gamma_{0}(j))\\
    \gamma & (\gamma_{1}(i) = \langle\Z\rangle, \gamma_{1}(j) = \langle\X\rangle, \mbox{ or}\\
    & \;\;\gamma_{1}(i) = \gamma_{1}(j) = \I)\\
    \gamma[\langle Z \rangle/i] & (\gamma_{1}(i) = \I)\\
    \gamma[\langle X \rangle/j] & (\gamma_{1}(j) = \I)\\
    \gamma[\trans{\appon{\CX}{i,j}}{(\gamma_{1}(i) \otimes \gamma_{1}(j))}/i,j] & (\mbox{otherwise})
  \end{array}
  \right.
\end{align*}
\begin{align*}
  \semEl{
    \begin{array}{rlc}
      \tif & i &\\
      & \tthen & C\\
      & \telse & C'\\
      \tfi & &
    \end{array}
  }(\gamma) &=
  \semE{C}(meas_{E}(i, \gamma)) \uplus \semE{C'}(meas_{E}(i, \gamma))\\
  \semE{\swhile{i}{C}}(\gamma) &=
  \biguplus_{n \in \N} meas_{E, i} ( {(\semE{C} \circ meas_{E, i})}^{n} (\gamma))
\end{align*}
where $U \in \{\tX, \tY, \tZ, \tH, \tS\}$.

The result $update_{E}(J, A, E)$ is computed as follows.
(1) If $E$ belongs to $\Sc$, then
$update(J, A, E)$ is the result.
(2) If not, take all $j_0, \ldots, j_{k-1} \subset J$ such that $E$ has rows $\appon{\diamondsuit_{j_l}}{j_l}$ where $\diamondsuit_{j_l} \in \{\X, \Y, \Z\}$.
When there is no such $j_l$, $\{(A, E)\}$ is the result.
Otherwise, define $K = A \backslash \{j_0, \ldots, j_{k-1}\}$.
Then, the result is $\{(K, \U)\} \cup \{ (\{j_l\},\diamondsuit_{j_l}) \mid l = 0, \ldots, k-1\}$.

The result of $meas_{E}(i, \gamma)$ varies with $\gamma$.
(1) If $|\gamma_{0}(i)| = 1$, then $\gamma[\langle \Z \rangle/i]$ is the result.
(2) If $\gamma_{1}(i)$ belongs to $\Sc \backslash \Sc[1]$, then $\gamma[update(\gamma_{0}(i), \gamma_{0}(i), meas_{st}(i, \gamma_{1}(i)))/i]$, which is the same as $meas$.
(3) If exactly one row of $\gamma_{1}(i)$ has $\X$ or $\Y$ in the $i$th column and the others have $\I$ or $\Z$, $meas_{E}(i, \gamma)$ is computed as follows.
First, the row and the $i$th column are removed from $\gamma_{1}(i)$.
Let us call the matrix $E'$.
Then $\gamma[\{(\{i\}, \langle \Z \rangle)\} \cup update_{E}(\gamma_{0}(i) \backslash \{i\}, \gamma_{0}(i) \backslash \{i\}, E')/i]$ is the result.
(4) Otherwise, we cannot obtain information about the post-measurement state.
The result is $\gamma[\{(\{i\}, \langle \Z \rangle), (\gamma_{0}(i) \backslash \{i\}, \U)\}/i]$.

The function $add_{\D}$ changes $\X$ and $\Y$ in the $i$th column into $\D$.
By the definition of equality in $\E$, we can assume that exactly one of the following holds:
(1) the $i$th column does not contain $\X$ or $\Y$,
(2) exactly one L-row has $\X$ or $\Y$ in the $i$th column, and
(3) only $\D$-rows have $\X$ or $\Y$ in the $i$th column.
In the first case, $add_{\D}$ does nothing and returns the second argument.
In the second and third cases, $add_{\D}$ changes all $\X$ and $\Y$ in the $i$th column into $\D$ and returns the matrix.
Hence, $add_{\D}$ changes at most one L-row into a $\D$-row.

\begin{exm}
Now, we compute $\semE{\texttt{exm}_{1}}(\gamma)$.
\begin{equation*}
  \left[
    \begin{array}{ccc}
      \X & \X & \Z\\
      \Z & \Z & \I\\
      \Z & \I & \Z
    \end{array}
    \right]
  \xrightarrow{\sT{1}}
  \left[
    \begin{array}{ccc}
      \X & \D & \X\\
      \Z & \Z & \I\\
      \Z & \I & \Z
    \end{array}
    \right]
  \xrightarrow{\texttt{meas}(0)}
  \left[
    \begin{array}{cccc}
      \multicolumn{1}{c|}{\Z} &  &\\
      \cline{1-2}
      & \multicolumn{1}{|c|}{\Z} & \\
      \cline{2-3}
      && \multicolumn{1}{|c}{\Z}
    \end{array}
    \right]
\end{equation*}
Thus, we conclude that all qubits are separated.
\end{exm}

Finally, we show $\semE{\cdot}$ is sound.
\begin{thm}
  For any state $\rho$, extended assignment $\gamma$, and program $C$, $\gamma \vDash \rho$ implies $\semE{C}(\gamma) \vDash \sem{C}(\rho)$.
\end{thm}
\begin{proof}
  By induction on the structure of $C$.
  For $\sskip$, $\scomp{C}{C'}$, $U(i)$, and $\sT{i}$, it is easy.
  For $\sCX{i}{j}$, since the number of $\D$-rows does not increase, the statement holds.
  Extended stabilisers also satisfy $\gamma \uplus \delta \vDash \rho + \sigma$ whenever $\gamma \vDash \rho$ and $\delta \vDash \sigma$.
  For $\sif{i}{C}{C'}$, we have to check $meas_{E}$.
  However, since it also just decrease the number of $\D$-rows, $meas_{E}(\gamma)\vDash \proj{\appon{\ketbra{0}{0}}{i}}{\rho}$.
  Hence, the statement holds for $\sif{i}{C}{C'}$.
  Finally, we show for $\swhile{i}{C}$.
  Since $C^{\Q}$ is finite, $\semE{\swhile{i}{C}}(\gamma) \vDash \sum_{n \leq M} \allowbreak \proj{\appon{\ketbra{1}{1}}{i}}{{f}^{n}(\rho)}$ for sufficiently large $M$.
  Therefore, the statement holds by continuity of projection.
\end{proof}

\section{Conclusion}
We used stabilizer formalism to improve entanglement analysis in quantum programs.
First, we introduced an abstract domain $C^{\Q}$ and an abstract semantics.
It assigns stabilizers or non-stabilizers to each segment of a quantum state, where non-stabilizers are assigned when non-Clifford gates are applied to the segment.
The method enables us to analyse separability of qubits in quantum programs more precisely.
Specifically, we could deduce that all qubits are separated after executing $\texttt{SEP}_{0}$ or $\texttt{SEP}_{1}$.
Moreover, we defined an abstract domain $E^{\Q}$, as $C^{\Q}$ introduces too many non-stabilizers.
Even when non-Clifford gates appear, the domain does not discard stabilizers but keeps Pauli matrices that are not disturbed by the gates.
Hence, it suppresses effects of non-Clifford gates that will be removed later.
We showed soundness of both semantics.

In a field of model checking, stabilizer formalism was used to verify quantum programs and analyse entanglement of those programs~\cite{GNP08, GNP09}.
However, in these studies, quantum gates in a target language were restricted to only Clifford gates.
It is worth noting that our target language QIL has a non-Clifford gate.
This is a big advantage of our work and actually one of the challenges of our work was how to manage the non-Clifford gate.
We restricted the effect by overapproximation.
Although we refined the approximation from $C^{\Q}$ to $E^{\Q}$, further refinement is still needed such as finding a better approximate join operator in $E^{\Q}$.

\section*{Acknowledgements}
We thank Yoshihiko Kakutani for helpful comments.
This work was supported by JSPS Grant-in-Aid for JSPS Fellows Grant No.\ $26\mathrel{\cdot}9148$.

\bibliographystyle{eptcs}
\bibliography{Paper66}

\begin{thebibliography}{10}
\providecommand{\bibitemdeclare}[2]{}
\providecommand{\surnamestart}{}
\providecommand{\surnameend}{}
\providecommand{\urlprefix}{Available at }
\providecommand{\url}[1]{\texttt{#1}}
\providecommand{\href}[2]{\texttt{#2}}
\providecommand{\urlalt}[2]{\href{#1}{#2}}
\providecommand{\doi}[1]{doi:\urlalt{http://dx.doi.org/#1}{#1}}
\providecommand{\bibinfo}[2]{#2}

\bibitemdeclare{article}{AG04}
\bibitem{AG04}
\bibinfo{author}{Scott \surnamestart Aaronson\surnameend} \&
  \bibinfo{author}{Daniel \surnamestart Gottesman\surnameend}
  (\bibinfo{year}{2004}): \emph{\bibinfo{title}{Improved simulation of
  stabilizer circuits}}.
\newblock {\sl \bibinfo{journal}{Physical Review A}} \bibinfo{volume}{70}, p.
  \bibinfo{pages}{052328}, \doi{10.1103/PhysRevA.70.052328}.

\bibitemdeclare{article}{AP05}
\bibitem{AP05}
\bibinfo{author}{Koenraad M~R \surnamestart Audenaert\surnameend} \&
  \bibinfo{author}{Martin~B \surnamestart Plenio\surnameend}
  (\bibinfo{year}{2005}): \emph{\bibinfo{title}{Entanglement on mixed
  stabilizer states: normal forms and reduction procedures}}.
\newblock {\sl \bibinfo{journal}{New Journal of Physics}}
  \bibinfo{volume}{7}(\bibinfo{number}{1}), p. \bibinfo{pages}{170},
  \doi{10.1088/1367-2630/7/1/170}.

\bibitemdeclare{article}{BBCJPW93}
\bibitem{BBCJPW93}
\bibinfo{author}{Charles~H. \surnamestart Bennett\surnameend},
  \bibinfo{author}{Gilles \surnamestart Brassard\surnameend},
  \bibinfo{author}{Claude \surnamestart Cr\'epeau\surnameend},
  \bibinfo{author}{Richard \surnamestart Jozsa\surnameend},
  \bibinfo{author}{Asher \surnamestart Peres\surnameend} \&
  \bibinfo{author}{William~K. \surnamestart Wootters\surnameend}
  (\bibinfo{year}{1993}): \emph{\bibinfo{title}{Teleporting an unknown quantum
  state via dual classical and {Einstein-Podolsky-Rosen} channels}}.
\newblock {\sl \bibinfo{journal}{Physical Review Letters}}
  \bibinfo{volume}{70}(\bibinfo{number}{13}), pp. \bibinfo{pages}{1895--1899},
  \doi{10.1103/PhysRevLett.70.1895}.

\bibitemdeclare{article}{BW92}
\bibitem{BW92}
\bibinfo{author}{Charles~H. \surnamestart Bennett\surnameend} \&
  \bibinfo{author}{Stephen~J. \surnamestart Wiesner\surnameend}
  (\bibinfo{year}{1992}): \emph{\bibinfo{title}{Communication via one- and
  two-particle operators on Einstein-Podolsky-Rosen states}}.
\newblock {\sl \bibinfo{journal}{Physical Review Letters}}
  \bibinfo{volume}{69}, pp. \bibinfo{pages}{2881--2884},
  \doi{10.1103/PhysRevLett.69.2881}.

\bibitemdeclare{incollection}{GNP08}
\bibitem{GNP08}
\bibinfo{author}{Simon~J. \surnamestart Gay\surnameend},
  \bibinfo{author}{Rajagopal \surnamestart Nagarajan\surnameend} \&
  \bibinfo{author}{Nikolaos \surnamestart Papanikolaou\surnameend}
  (\bibinfo{year}{2008}): \emph{\bibinfo{title}{QMC: A Model Checker for
  Quantum Systems}}.
\newblock In \bibinfo{editor}{Aarti \surnamestart Gupta\surnameend} \&
  \bibinfo{editor}{Sharad \surnamestart Malik\surnameend}, editors: {\sl
  \bibinfo{booktitle}{Computer Aided Verification}}, {\sl
  \bibinfo{series}{Lecture Notes in Computer Science}} \bibinfo{volume}{5123},
  \bibinfo{publisher}{Springer Berlin Heidelberg}, pp.
  \bibinfo{pages}{543--547}, \doi{10.1007/978-3-540-70545-1\_51}.

\bibitemdeclare{incollection}{GNP09}
\bibitem{GNP09}
\bibinfo{author}{Simon~J. \surnamestart Gay\surnameend},
  \bibinfo{author}{Rajagopal \surnamestart Nagarajan\surnameend} \&
  \bibinfo{author}{Nikolaos \surnamestart Papanikolaou\surnameend}
  (\bibinfo{year}{2009}): \emph{\bibinfo{title}{Specification and Verification
  of Quantum Protocols}}.
\newblock In \bibinfo{editor}{Simon \surnamestart Gay\surnameend} \&
  \bibinfo{editor}{Ian \surnamestart Mackie\surnameend}, editors: {\sl
  \bibinfo{booktitle}{Semantic Techniques in Quantum Computation}},
  \bibinfo{publisher}{Cambridge University Press}, pp.
  \bibinfo{pages}{414--472}, \doi{10.1017/CBO9781139193313.012}.
\newblock \bibinfo{note}{Cambridge Books Online}.

\bibitemdeclare{article}{G96}
\bibitem{G96}
\bibinfo{author}{Daniel \surnamestart Gottesman\surnameend}
  (\bibinfo{year}{1996}): \emph{\bibinfo{title}{Class of quantum
  error-correcting codes saturating the quantum Hamming bound}}.
\newblock {\sl \bibinfo{journal}{Physical Review A}} \bibinfo{volume}{54}, pp.
  \bibinfo{pages}{1862--1868}, \doi{10.1103/PhysRevA.54.1862}.

\bibitemdeclare{incollection}{JP09}
\bibitem{JP09}
\bibinfo{author}{Philippe \surnamestart Jorrand\surnameend} \&
  \bibinfo{author}{Simon \surnamestart Perdrix\surnameend}
  (\bibinfo{year}{2009}): \emph{\bibinfo{title}{Abstract Interpretation
  Techniques for Quantum Computation}}.
\newblock In \bibinfo{editor}{Simon \surnamestart Gay\surnameend} \&
  \bibinfo{editor}{Ian \surnamestart Mackie\surnameend}, editors: {\sl
  \bibinfo{booktitle}{Semantic Techniques in Quantum Computation}},
  \bibinfo{publisher}{Cambridge University Press}, pp.
  \bibinfo{pages}{206--234}, \doi{10.1017/CBO9781139193313.007}.

\bibitemdeclare{book}{NC00}
\bibitem{NC00}
\bibinfo{author}{Michael~A. \surnamestart Nielsen\surnameend} \&
  \bibinfo{author}{Isaac~L. \surnamestart Chuang\surnameend}
  (\bibinfo{year}{2000}): \emph{\bibinfo{title}{Quantum Computation and Quantum
  Information}}.
\newblock \bibinfo{publisher}{Cambridge University Press}.

\bibitemdeclare{article}{P07}
\bibitem{P07}
\bibinfo{author}{Simon \surnamestart Perdrix\surnameend}
  (\bibinfo{year}{2007}): \emph{\bibinfo{title}{Quantum Patterns and Types for
  Entanglement and Separability}}.
\newblock {\sl \bibinfo{journal}{Electronic Notes in Theoretical Computer
  Science}} \bibinfo{volume}{170}(\bibinfo{number}{0}), pp.
  \bibinfo{pages}{125--138}, \doi{10.1016/j.entcs.2006.12.015}.

\bibitemdeclare{article}{P08a}
\bibitem{P08a}
\bibinfo{author}{Simon \surnamestart Perdrix\surnameend}
  (\bibinfo{year}{2008}): \emph{\bibinfo{title}{A Hierarchy of Quantum
  Semantics}}.
\newblock {\sl \bibinfo{journal}{Electronic Notes in Theoretical Computer
  Science}} \bibinfo{volume}{192}(\bibinfo{number}{3}), pp.
  \bibinfo{pages}{71--83}, \doi{10.1016/j.entcs.2008.10.028}.

\bibitemdeclare{incollection}{P08b}
\bibitem{P08b}
\bibinfo{author}{Simon \surnamestart Perdrix\surnameend}
  (\bibinfo{year}{2008}): \emph{\bibinfo{title}{Quantum Entanglement Analysis
  Based on Abstract Interpretation}}.
\newblock In \bibinfo{editor}{Mar\'{i}a \surnamestart Alpuente\surnameend} \&
  \bibinfo{editor}{Germ\'{a}n \surnamestart Vidal\surnameend}, editors: {\sl
  \bibinfo{booktitle}{Static Analysis}}, {\sl \bibinfo{series}{Lecture Notes in
  Computer Science}} \bibinfo{volume}{5079}, \bibinfo{publisher}{Springer
  Berlin Heidelberg}, pp. \bibinfo{pages}{270--282},
  \doi{10.1007/978-3-540-69166-2\_18}.

\bibitemdeclare{incollection}{PZ09}
\bibitem{PZ09}
\bibinfo{author}{Fr\'{e}d\'{e}ric \surnamestart Prost\surnameend} \&
  \bibinfo{author}{Chaouki \surnamestart Zerrari\surnameend}
  (\bibinfo{year}{2009}): \emph{\bibinfo{title}{Reasoning about Entanglement
  and Separability in Quantum Higher-Order Functions}}.
\newblock In \bibinfo{editor}{CristianS. \surnamestart Calude\surnameend},
  \bibinfo{editor}{Jos\'{e}F\'{e}lix \surnamestart Costa\surnameend},
  \bibinfo{editor}{Nachum \surnamestart Dershowitz\surnameend},
  \bibinfo{editor}{Elisabete \surnamestart Freire\surnameend} \&
  \bibinfo{editor}{Grzegorz \surnamestart Rozenberg\surnameend}, editors: {\sl
  \bibinfo{booktitle}{Unconventional Computation}}, {\sl
  \bibinfo{series}{Lecture Notes in Computer Science}} \bibinfo{volume}{5715},
  \bibinfo{publisher}{Springer Berlin Heidelberg}, pp.
  \bibinfo{pages}{219--235}, \doi{10.1007/978-3-642-03745-0\_25}.

\bibitemdeclare{article}{RB01}
\bibitem{RB01}
\bibinfo{author}{Robert \surnamestart Raussendorf\surnameend} \&
  \bibinfo{author}{Hans~J. \surnamestart Briegel\surnameend}
  (\bibinfo{year}{2001}): \emph{\bibinfo{title}{A One-Way Quantum Computer}}.
\newblock {\sl \bibinfo{journal}{Physical Review Letters}}
  \bibinfo{volume}{86}, pp. \bibinfo{pages}{5188--5191},
  \doi{10.1103/PhysRevLett.86.5188}.

\bibitemdeclare{article}{S04}
\bibitem{S04}
\bibinfo{author}{Peter \surnamestart Selinger\surnameend}
  (\bibinfo{year}{2004}): \emph{\bibinfo{title}{Towards a quantum programming
  language}}.
\newblock {\sl \bibinfo{journal}{Mathematical Structures in Computer Science}}
  \bibinfo{volume}{14}, pp. \bibinfo{pages}{527--586},
  \doi{10.1017/S0960129504004256}.

\end{thebibliography}

\nocite{JP09}
\end{document}